\newcommand{\jo}[1]{}
\newcommand{\si}[1]{#1}

\RequirePackage{fix-cm}
\jo{
\documentclass[smallextended]{svjour3}       
}
\si{
\documentclass{article}
}
\jo{
\smartqed  
}
\usepackage{comment}
\usepackage{pagecolor,lipsum}
\usepackage{mathrsfs}
\usepackage{enumitem}
\usepackage{graphicx}
\usepackage{tikz}
\usepackage[normalem]{ulem}
\usepackage[misc]{ifsym}

\si{
\usepackage{amsthm}
\newtheorem{theorem}{Theorem}

\newtheorem{lemma}{Lemma}
\newtheorem{proposition}{Proposition}
\newtheorem{algorithm}{Algorithm}
\theoremstyle{definition}
\newtheorem{definition}{Definition}

\theoremstyle{remark}
\newtheorem{remark}{Remark}
}
\usepackage{amsmath,amsfonts,amssymb}
\usepackage{subfigure}
\usepackage[T1]{fontenc}
\usepackage[normalem]{ulem}
\usepackage{textcomp}
\usepackage{algpseudocode}
\usepackage{algorithm}
\usepackage{multirow}
\usepackage{framed}
\DeclareMathAlphabet{\mathpzc}{OT1}{pzc}{m}{it}

\usepackage{color}
\newcommand{\R}{{\mathbb R}}
\jo{
\journalname{Mathematical Programming}
}

\begin{document}
\title{Optimality condition and complexity analysis for linearly-constrained optimization without differentiability on the boundary
}

\jo{
\titlerunning{Optimality condition and complexity analysis for non-smooth optimization}        
}

\jo{
\author{
Gabriel Haeser  \and
Hongcheng Liu        \and
        Yinyu Ye 
}
}

\si{
\author{
G. Haeser\thanks{
Department of Applied Mathematics, University of S\~ao Paulo, S\~ao Paulo SP, Brazil. Visiting Scholar at Department of Management Science and Engineering, Stanford University, Stanford CA 94305, USA. E-mail: ghaeser@ime.usp.br.}
\and
Hongcheng Liu\thanks{Department of Radiation Oncology, Stanford University, Stanford CA 94305, USA. E-mail: hql5143liu@gmail.com}
              \and
              Yinyu Ye\thanks{Department of Management Science and Engineering, Stanford University, Stanford CA 94305, USA. E-mail: yinyu-ye@stanford.edu}
}
}


\jo{
\institute{G. Haeser \Letter \at
Department of Applied Mathematics, 
\\Institute of
    Mathematics and Statistics, University of S\~ao Paulo, S\~ao Paulo SP,
    Brazil. \\
 \email{ghaeser@ime.usp.br} 
 \and
H. Liu \at
                 Department of Radiation Oncology, 
              \\Stanford University, Stanford, CA 94305, USA\\
              \email{hql5143liu@gmail.com}           
              \and
              Y. Ye \at
              Department of Management Science and Engineering, \\
              Stanford University, Stanford, CA 94305, USA\\
              \email{yinyu-ye@stanford.edu}     
}
}

\jo{
\date{Received: date / Accepted: date}
}
\si{
\date{\today}
}

\maketitle

\begin{abstract}
In this paper we consider the minimization of a continuous function that is potentially not differentiable or not twice differentiable on the boundary of the feasible region. By exploiting an interior point technique, we present first- and second-order optimality conditions for this problem that reduces to classical ones when the derivative on the boundary is available. For this type of problems, existing necessary conditions often rely on the notion of subdifferential or become non-trivially weaker than the KKT condition in the (twice-)differentiable counterpart problems. In contrast, this paper  presents a new set of first- and second-order necessary conditions that are derived without the use of subdifferential and reduces to exactly the KKT condition when (twice-)differentiability holds.  As a result, these conditions are stronger than some existing ones considered for the discussed minimization  problem when only non-negativity constraints are present. To solve for these optimality conditions in the special but important case of linearly constrained problems, we  present two novel interior trust-region point algorithms and show that their worst-case computational efficiency in achieving the potentially  stronger optimality conditions match the best known complexity bounds. Since this work considers a more general problem  than the literature, our results also indicate that best known complexity bounds hold for a wider class of nonlinear programming problems.



\jo{
}
\si{
}
\jo{
  
 }
\end{abstract}

\section{Introduction}

In this paper we are interested in the problem

\begin{equation}
\label{linp}
\begin{array}{ll}\mbox{Minimize}&f(x),\\
\mbox{subject to}&\mathbf Ax=\mathbf b, x\geq 0,\end{array}\end{equation}
where $\mathbf A\in\R^{m\times n}$ and $f:\R^n_+\to\R$ is a continuous function on $\R^n_+:=\{x\in\R^n\mid x\geq 0\}$ and smooth on $\R^n_{++}:=\{x\in\R^n\mid x>0\}$. As a special case of \eqref{linp}, the following formulation has been popularly studied:
\begin{equation}
\label{boxp}\begin{array}{ll}\mbox{Minimize }&H(x)+\lambda\sum_{i=1}^n\varphi(x_i^p),\\
\mbox{subject to }&x\geq0,\end{array}\end{equation}
where $H$ is smooth, $\varphi$ is convex, $\lambda>0$ and $0<p<1$. A common use of \eqref{boxp} (or its immediate reformulations) is the problem of high-dimensional learning under the assumption of sparsity. In such a problem, few data observations are acquired for the task of recovering a high-dimension signal. Such a task is often done by minimizing an in-sample statistical loss (a.k.a., fidelity) function  $H(x)$ that represents the in-sample  error plus a regularization function $\lambda\sum_{i=1}^n \varphi(x_i^p)$, which penalizes non-zero variables to induce sparsity. Theoretical and numerical studies on the efficacies of this type of models are presented in \cite{Negahbanetal,liu,FanandLi2001,FanandLv2011,FanLvQireview,Fanetal2012,LohandWainwright,Wangetal2013ultrahigh,Wangetal2013}. Particularly, it is shown by \cite{liu,LohandWainwright,Wangetal2013,Fanetal2012,FanandLv2011,Wangetal2013ultrahigh} that to achieve a sound recovery quality, global optimality to \eqref{linp} is not necessary, but some local minima or even stationary points can successfully recover the high-dimensional signal with high probability. In specific, \cite{liu} shows that solutions satisfying a second-order necessary condition in  linear regression penalized by certain nonconvex $\varphi(x_i^p)$ have very desirable statistical properties.  \cite{deepcompression} presented a recent application of \eqref{boxp} in designing neural networks for deep learning, for which $\varphi(x_i^p)=\vert x\vert$ or $\varphi(x_i^p)=\Vert x\Vert^2$ and $H$ is a nonconvex loss function. 

Despite various successful and seminal applications,  \eqref{boxp} remains a non-trivial problem to solve due to the usual  absence of differentiability or twice-differentiability and the frequent presence of nonconvexity. 
As an example, if $p<1$, the function  {$\sum_{i=1}^n x_i^p$} is not even directionally differentiable in G\^{a}teaux sense when $x_i=0$ for any $i$. Similarly, when $p<2$, the objective function is not twice differentiable. Meanwhile, in the training of a neural network, $H$ is usually smooth but nonconvex, as in the case of \cite{deepcompression}. \cite{Wangetal2013} discussed some other cases where $H$ is nonconvex.

To establish first-/second-order necessary optimality conditions for local minimality, different variants of the KKT condition have been discussed  {when differentiability is potentially absent. In such a case,  optimality conditions based on the notion of subdifferential are studied by \cite{chenpenalty,Variational analysis,Audet,Jahn}.} Weaker optimality conditions without the use of subdifferential have been discussed by \cite{biansmooth,bian,bianlinear,liuma}. Interested readers are referred to \cite{bianMOR} for an excellent review on the optimality conditions. In particular, \cite{bian} considers the so-called scaled first-order optimality condition for \eqref{boxp}:
\begin{align}
x_i\frac{\partial{H(x)}}{\partial x_i}+\lambda p\varphi'(x_i^p)x_i^p=0,\qquad\forall i=1,\dots,n.\label{scaled KKT}
\end{align}
This condition is evidently weaker than the  conditions by \cite{chenpenalty,Variational analysis,Audet,Jahn}, in that \eqref{scaled KKT} always holds at the origin regardless of the objective function. According to \cite{bianMOR}, similar issues apply to the optimality conditions in \cite{biansmooth,bianlinear,liuma}.  {In contrast, our presented optimality condition does not rely on any form of subdifferential and is equivalent to the canonical version of the KKT condition when $f$ is smooth. Therefore, the presented optimality condition is tighter than \cite{biansmooth,bian,bianlinear,liuma}.}

 {Our research is also motivated by the need of characterizing approximations to the ``exact'' necessary condition, since it is generally impossible to solve \eqref{linp} exactly, even only for KKT solutions.}  As a result, the ``exact''  first- or second-order necessary conditions must be perturbed to properly characterize the actual solution  {obtained} through an algorithm. Furthermore, it is desirable to establish a connection between the optimality condition and its $\varepsilon$ perturbed version  {(approximation with inaccuracy measured by $\varepsilon$)} in order for the complexity results to be meaningful. 
 {Approximate KKT-like conditions in solving nonconvex and nonsmooth optimization have been proposed by \cite{bian,bianlinear,chenpenalty,bianMOR}.  In view of this gap in the literature, this paper presents  a set of  perturbed (first- and second-order) necessary optimality conditions that are originally defined in terms of a limit of perturbed stationary points.  Compare to \cite{chenpenalty,bianMOR}, our perturbed  necessary conditions are  free from the use of subdifferential, and  are stronger than \cite{bian,bianlinear}.}



To compute solutions satisfying our proposed perturbed necessary conditions, we develop a first- and second-order interior trust-region point (ITRP) algorithms. Both algorithms work in a general setting  that allows for irregularities of the objective function unaddressed in the literature. In particular, the first-order ITRP allows $f$ to be not even directionally differentiable. The resulting computational complexity, $O(\varepsilon^{-2})$ in achieving an $\varepsilon$-perturbed first-order stationary point (where $\varepsilon>0$), coincides with the best known complexity for solving smooth nonconvex problems using only first-order information and assuming the absence of matrix inversion. The second-order ITRP then applies to a class of problems where second-order derivative may not exist. The resulting complexity, $O(\varepsilon^{-3/2})$ and $O(\varepsilon^{-3})$ in achieving an $\varepsilon$-perturbed first-order and second-order stationary point, respectively, equals the best-known complexity for twice continuously differentiable functions. The corresponding $\varepsilon$-perturbed necessary optimality conditions are in stronger forms than those discussed in  {\cite{bian,bianlinear,chenpenalty,bianMOR}}. We further show that, at the same rate of complexity, the same type of  $\varepsilon$-perturbed  scaled optimality condition as in \cite{bian} can be achieved for a more general set of optimization problems by our second-order ITRP. For a comprehensive analysis of the ITRP, we further considered the case where $f$ is a quadratic function and present an alternative analysis for the same result in \cite{Ye98}. In such a special case, the ITRP is substantially accelerated and achieves both the first- and second-order conditions at a rate of $O(\varepsilon^{-1})$.



In contrast, in the literature, for smooth unconstrained optimization, when only first-order information is accessible and no matrix inversion is involved, the algorithms with best known complexity bounds take at most $O(\varepsilon^{-2})$ iterations to achieve a first-order stationary point up to a tolerance $\varepsilon$. It is the case of the steepest descent \cite{nesterovbook}, trust region methods \cite{trustregion} and the nonlinear stepsize control algorithms \cite{grapiglia2,tointnsc}, for instance. When second-order derivatives are used, the best known complexity is reduced to $O(\varepsilon^{-3/2})$ for first-order stationarity and, to find a second-order stationary point perturbed by $\varepsilon$, the best known complexity is $O(\varepsilon^{-3})$. See   \cite{grapiglia2,tointnsc,birginquad,curtis,tointsecond,tointsecond2,Nesterov-Polyak,martinez}. A different line of reasoning appeared recently in \cite{oliver,agar}, where the second-order information is iteratively approximated by the first-order one. In this case, the complexity bound of $O(\varepsilon^{-7/4})$ can be achieved for first-order stationarity. We do not pursue this last type of results. The best complexity bounds known are the same if constraints are considered \cite{tointconst,tointconst2} or in some nonsmooth cases \cite{biansmooth,bian,bianlinear,tointcomposite,grapiglia}. Our algorithms will achieve the best known complexity bounds of $O(\varepsilon^{-2})$, $O(\varepsilon^{-3/2})$ and $O(\varepsilon^{-3})$, depending on the use of second-order information. To our knowledge, our problem of discussion is more general than most existing developments in the literature.

The rest of the paper is organized in the following way. Section 2 articulates  our optimality condition and  Section 3 presents our algorithm and  complexity analyses. Finally,  Section 4 concludes the paper.\\

{\bf Notation.} Given $n\geq1$, $\R^n_+$ is the non-negative orthant in $\R^n$. We denote by $\R^n_{++}\subset\R^n_+$ the subset of vectors with all coordinates positive. Given $x\in\R^n$, we denote $diag(x)$ the diagonal matrix defined by $x$. When it is clear from confusion, we call $X=diag(x)$. The vectors $e_1,\dots,e_n$ is the canonical basis of $\R^n$ and $e\in\R^n$ is the vector of ones. The identity matrix of appropriate dimention will be denoted $\mathcal{I}$. Given a symmetric matrix $A$, we denote by $A\succeq 0$ when $A$ is positive semidefinite. The gradient vector and hessian matrix of a function $f:\R^n\to\R$ at $x\in\R^n$ is denoted, respectively, by $\nabla f(x)$ and $\nabla^2 f(x)$. We use $\|\cdot\|$ and $\|\cdot\|_\infty$ to represent the $\ell_2$- and $\ell_\infty$-norms, respectively. The smallest integer greater than or equal to $x\in\R$ is denoted by $\lceil x\rceil$.

\section{Optimality condition}

Let us consider, for simplicity, a special case of \eqref{linp} with only bound constraints $x\geq0$ and let us assume that for each $i=1,\dots,n$, the partial derivative $\frac{\partial f(x)}{\partial x_i}$ is not defined when $x_i=0$. A so-called scaled first-order optimality condition holds at a local minimizer $x^*$, given by $x^*_i\frac{\partial f(x^*)}{\partial x_i}=0, i=1,\dots,n$, where the product is taken to be zero when the derivative does not exist. See \cite{chenye}.

A point $x>0$ with $|x_i\frac{\partial f(x)}{\partial x_i}|\leq\varepsilon$ for all $i=1,\dots,n$, is called an $\varepsilon$-scaled first-order point. See \cite{bian}. In \cite{bianlinear}, it was proved that if a sequence $\{x^k\}\subset\R^n$ is such that $x^k\to x^*$ and $x^k$ is an $\varepsilon_k$-scaled first-order point for all $k$ with some $\varepsilon_k\to0^+$, then $x^*$ is a scaled first-order point. Combining both results, the situation is the one described in Figure \ref{bian}. Algorithms thus proceed to find $\varepsilon$-scaled first-order points, with some small $\varepsilon>0$ as in \cite{bian,bianlinear,liuma}.

\begin{figure}[h]
    \begin{center}
        \begin{tikzpicture}
            \draw[black, ultra thick] (-3,0) ellipse (2 cm and 1.5 cm);
            \draw[black, ultra thick] (-0,0) ellipse (2 cm and 1.5 cm);
            \draw[black, ultra thick] (-6,-2) rectangle (3,2);
            \node [above] at (0,0.2) {Limits of};
            \node [above] at (0.1,-0.25) {$\varepsilon$-scaled};
            \node [above] at (0.3,-0.8) {points ($\varepsilon\to0^+$)};
            \node [above] at (-3,0) {Local};
             \node [above] at (-3,-0.5) {minimizers};
            \node [above] at (1,-2) {Scaled first-order points};
        \end{tikzpicture}
    \end{center}
    \caption{Local minimizers and limits of $\varepsilon$-scaled first-order points, $\varepsilon\to0^+$, are scaled first-order points. Since a scaled first-order point can be seen as a weak necessary optimality condition, this gives little theoretical justification for considering an $\varepsilon$-scaled first-order point, $\varepsilon>0$, as an approximate solution.}
    \label{bian}
\end{figure}
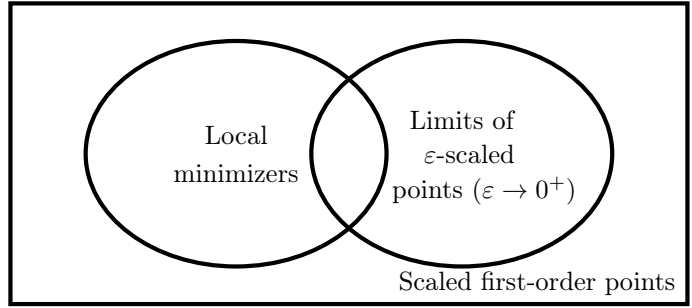

A first issue with this approach is that there is no analogous of the condition $\nabla f(x)\geq0$, present in the canonical KKT conditions when derivatives exist everywhere. This is overcome in \cite{bian,bianlinear,liuma} by considering the particular objective function \eqref{boxp}, where $\frac{\partial f(x)}{\partial x_i}\to+\infty$ when $x_i\to0$, or considering an optimality condition based on the computation of subdifferentials \cite{bianMOR}. A second issue is the fact that there is no measure of strength of the scaled first-order optimality condition, since, for instance, it always holds at $x=0$, regardless of the objective function. Finally, a third issue is the lack of relation between local minimizers and limits of $\varepsilon$-scaled first-order points, as suggested by Figure \ref{bian}. A similar criticism apply to the scaled second-order condition considered in \cite{bian}, and other first-order optimality conditions considered for this class of problems. See \cite{bianMOR} and references therein.

We will overcome these issues by defining first- and second-order optimality conditions that coincide with the canonical first- and second-order KKT conditions under usual smoothness assumptions, in a much more general framework. The optimality condition is defined in such a way that it naturally suggest an $\varepsilon$ perturbed first- and second-order criterion suitable for the complexity analysis. We also show that, in the case of linear constraints, our first-order (second-order) optimality condition can be satisfied by the computation of $\varepsilon$-scaled first-order (second-order, respectively) points, as long as a suitable non-negativity criterion associated with the gradient of the objective function is fulfilled.

\subsection{Necessary Optimality Conditions Based on Limits of Perturbations}

This section presents optimality conditions for a much more general problem than \eqref{linp}. Specifically,  we consider the problem:
\begin{equation}
\label{genp}
\begin{array}{ll}\mbox{Minimize}&f(x),\\
\mbox{subject to}&h(x)=0, c(x)\geq 0,\end{array}\end{equation}
where, $f:\R^n\to\R, h:\R^n\to\R^m$ and $c:\R^n\to\R^p$. Defining $C^\circ:=\{x\mid c(x)>0\}$ and $C:=\{x\mid c(x)\geq0\}$, $f$, $h$ and $c$ are assumed to be continuous on $C$ and differentiable on $C^\circ$. For the second-order optimality condition, we assume also second-order differentiability on $C^\circ$. For any local solution $x^*$ of \eqref{genp}, assume that there exists a sequence $\{z^k\}$ with $z^k\to x^*$ and $z^k\in C^\circ\cap\{x\mid h(x)=0\}$ for all $k$, which is typically necessary for the application of interior point methods.  Also assume that for any point $x\in C^\circ\cap\{x\mid h(x)=0\}$, the rank of $\{\nabla h_i(y)\}_{i=1}^m$ is constant for all $y$ in a neighborhood of $x$.

Note that derivatives of objective function and constraints may not exist when some $c_i(x)=0$. Note also that we do not assume any constraint qualification on the whole feasible set.\\

\begin{theorem}\label{opt}Under the assumptions described above, let $x^*$ be a local solution of \eqref{genp}. Then, there exists a sequence of approximate solutions $\{x^k\}\subset\R^n$ and sequences of approximate Lagrange multipliers $\{\lambda^k\}\subset\R^m$, $\{s^k\}\subset\R^p_+$ such that:
\begin{enumerate}
\item[i)] $c(x^k)>0$, $h(x^k)=0$ for all $k$ and $x^k\to x^*$,
\item[ii)] $\nabla f(x^k)+\sum_{i=1}^m\lambda_i^k\nabla h_i(x^k)-\sum_{i=1}^ps_i^k\nabla c_i(x^k)\to0$,
\item[iii)] 
 $c_i(x^k)s_i^k\to0$ for all $i=1,\dots,p$.
\end{enumerate}
If, in addition, $f$, $h$, and $c$ are twice differentiable on $C^\circ$, then, there exist sequences $\{\theta^k\}\subset\R^p_+$ and $\{\delta_k\}\subset\R_+, \delta_k\to0^+$ such that
\begin{enumerate}
\item[iv)] $ d^\top(\nabla^2f(x^k)+\sum_{i=1}^m\lambda_i^k\nabla^2h_i(x^k)-\sum_{i=1}^ps_i^k\nabla^2c_i(x^k)+\sum_{i=1}^p\theta_i^k\nabla c_i(x^k)\nabla c_i(x^k)^\top+\delta_k\mathcal{I})d\geq0$, for all $d\in\R^n$ with $\nabla h_i(x^k)^\top d=0, i=1,\dots,m.$
\item[v)] 
$c_i(x^k)^2\theta_i^k\to0^+$ for all $i=1,\dots,p$.
\end{enumerate}\end{theorem}

\begin{proof} Let us take $\delta>0$ small enough such that the problem

\begin{equation}\label{reg}\mbox{Minimize }f(x)+\frac{1}{4}\|x-x^*\|^4\mbox{, s.t. }c(x)\geq0, h(x)=0, \|x-x^*\|^2\leq\delta,\end{equation}

has $x^*$ as its unique global solution. \\ 

Let us consider the application of the classical interior penalty method \cite{fiacco} to problem \eqref{reg} in the following sense: given a sequence $\{\mu_k\}\subset\R_+, \mu_k>0$ with $\mu_k\to0^+$, consider for every $k$ the problem:

\begin{equation}
\label{ipproblem}
\begin{array}{ll}
\mbox{Minimize}&\varphi_k(x):=f(x)+\frac{1}{4}\|x-x^*\|^4-\mu_k\sum_{i=1}^m\log(c_i(x)),\\
\mbox{subject to}&c(x)>0, h(x)=0, \|x-x^*\|^2\leq\delta.\end{array}\end{equation}

It is well known that a global solution $x^k$ exists for all $k$ and that cluster points of $\{x^k\}$ are global solutions of $\eqref{reg}$, see \cite{fiacco}. By the last constraint of $\eqref{ipproblem}$, $\{x^k\}$ is bounded, which implies that $x^k\to x^*$ and thus i) holds.\\

For $k$ large enough, $x^k$ is a local solution of 
\begin{equation*}\mbox{Minimize }\varphi_k(x):=f(x)+\frac{1}{4}\|x-x^*\|^4-\mu_k\sum_{i=1}^m\log(c_i(x))\mbox{, s.t. }h(x)=0.\end{equation*}

Since the constraints $h(x)=0$ satisfy a constraint qualification, there exist Lagrange multipliers $\lambda^k\in\R^m$ such that

\begin{align}0=&\nabla\varphi_k(x^k)+\sum_{i=1}^m\lambda_i^k\nabla h_i(x^k)\nonumber
\\=&\nabla f(x^k)+\|x^k-x^*\|^2(x^k-x^*)+\sum_{i=1}^m\lambda_i^k\nabla h_i(x^k)-\sum_{i=1}^p\frac{\mu_k}{c_i(x^k)}\nabla c_i(x^k),\nonumber
\end{align}

which gives ii) and iii) for $s^k_i:=\frac{\mu_k}{c_i(x^k)}, i=1,\dots,p$.

The second-order differentiability assumption and the constant rank condition around $x^k$ is enough to ensure that (see \cite{conjnino}):
\begin{align}
0\leq&\,d^\top(\nabla^2 \varphi(x^k)+\sum_{i=1}^m\lambda_i^k\nabla^2 h_i(x^k))d\nonumber
\\=&\,
d^\top\left(\nabla^2 f(x^k)+\sum_{i=1}^m\lambda_i^k\nabla^2 h_i(x^k)-\sum_{i=1}^ps_i^k\nabla^2 c_i(x^k)\right.\nonumber
\\&\left.+\sum_{i=1}^p\frac{\mu_k}{c_i(x^k)^2}\nabla c_i(x^k)\nabla c_i(x^k)^\top+2(x^k-x^*)(x^k-x^*)^\top+\|x^k-x^*\|^2\mathcal{I}\right)d,\nonumber
\end{align}
for all $d\in\R^n$ such that $\nabla h_i(x^k)^\top d=0, i=1,\dots,m$.

The result follows defining $\theta_i^k:=\frac{\mu_k}{c_i(x^k)^2}$ for all $i=1,\dots,p$, and $\delta^k\geq0$ as the largest eigenvalue of $2(x^k-x^*)(x^k-x^*)^\top+\|x^k-x^*\|^2\mathcal{I}$ for all $k$, which converges to zero.
\end{proof}

The optimality conditions immediately suggests definitions for $\varepsilon$-perturbed first- and second-order stationary points:\\

\begin{definition}\label{epskkt} Given $\varepsilon>0$, a point $x\in\R^n$ is called an $\varepsilon$-KKT point for problem \eqref{genp} when there exist approximate Lagrange multipliers $\lambda\in\R^m$ and $s\in\R^p_+$ with:
\begin{enumerate}
\item[(i)] $h(x)=0$, $c(x)>0$,
\item[(ii)] $\|\nabla f(x)+\sum_{i=1}^m\lambda_i\nabla h_i(x)-\sum_{i=1}^p s_i\nabla c_i(x)\|_{\infty}\leq\varepsilon$,
\item[(iii)] 
$|c_i(x)s_i|\leq\varepsilon$ for all $i=1,\dots,p$.
\end{enumerate}
\end{definition}

\begin{definition}\label{epskkt2}Given $\varepsilon>0$, a point $x\in\R^n$ is called an $\varepsilon$-KKT2 point for problem \eqref{genp} when there exist approximate Lagrange multipliers $\lambda\in\R^m$ and $s\in\R^p_+$ and a parameter $\theta\in\R^p_+$ with:
\begin{enumerate}
\item[(i)] $h(x)=0$, $c(x)>0$,
\item[(ii)] $\|\nabla f(x)+\sum_{i=1}^m\lambda_i\nabla h_i(x)-\sum_{i=1}^p s_i\nabla c_i(x)\|_{\infty}\leq\varepsilon$,
\item[(iii)] 
$|c_i(x)s_i|\leq\varepsilon$ for all $i=1,\dots,p$,
\item[(iv)] $d^\top\left(\nabla^2f(x)+\sum_{i=1}^m\lambda_i\nabla^2h_i(x)-\sum_{i=1}^ps_i\nabla^2c_i(x)+\sum_{i=1}^p\theta_i\nabla c_i(x)\nabla c_i(x)^\top+\varepsilon\mathcal{I}\right)d\geq0,$ for all $d\in\R^n$ with $\nabla h_i(x)^\top d=0, i=1,\dots,m,$
\item[(v)] 
$|c_i(x)^2\theta_i|\leq\varepsilon$ for all $i=1,\dots,p$.
\end{enumerate}
\end{definition}

Note that our first- and second-order optimality conditions given by Theorem \ref{opt} can be equivalently stated as, for all $\varepsilon>0$, there exist $\varepsilon$-KKT and, respectively, $\varepsilon$-KKT2 points, arbitrarily close to $x^*$.\\

The first-order optimality condition is the generalization of the ones from \cite{akkt,cakkt} to non-differentiable problems. In the smooth case, it implies the  {canonical}  first-order KKT conditions under weak constraint qualifications (see \cite{rcpld,cpg,ccp}), in particular, under linear constraints. The second-order optimality condition is the generalization of the one from \cite{akkt2,cakkt2} to the non-differentiable case and it implies the canonical second-order KKT conditions defined in terms of the critical subspace under weak constraint qualifications, in particular, under linear constraints. When the constraints are smooth, a formulation of the optimality condition in terms of perturbed critical directions is presented in \cite{auglag2}. 
We note that the results from \cite{cakkt2} can also be generalized without assuming smoothness on the boundary of $C$. In particular, without proving feasibility of the sequence $\{x^k\}$, the constant rank assumption can be dropped. 

\subsection{Sufficient Conditions for $\varepsilon$-Perturbed Stationary Points}

Let us now focus on a special case of \eqref{genp}, where we assume $h(x):=\mathbf Ax-\mathbf b$ and $c(x):=x$. This section then presents sufficient conditions for $\varepsilon$-KKT and $\varepsilon$-KKT2 points as per Definitions \ref{epskkt} and \ref{epskkt2}.\\

\begin{proposition}\label{epsilonKKT}Given $\varepsilon>0$, a sufficient condition for a point $x\in\R^n$ to be an $\varepsilon$-KKT point for problem \eqref{linp} is the existence of $\lambda\in\R^m$ such that:
\begin{enumerate}
\item[(a)] $\mathbf Ax=\mathbf b, x>0$,
\item[(b)] $\nabla f(x)+\mathbf A^\top\lambda\geq-\varepsilon$,
\item[(c)] $\|X(\nabla f(x)+\mathbf A^\top\lambda)\|_\infty\leq\varepsilon$.
\end{enumerate}
\end{proposition}
\begin{proof}Define $s:=\max\{0,\nabla f(x)+\mathbf A^\top\lambda\}$ in Definition \ref{epskkt} and the claimed result follows from an easy calculation.\end{proof}

\begin{proposition} Given $\varepsilon>0$, a sufficient condition for a point $x\in\R^n$ to be an $\varepsilon$-KKT2 point for problem \eqref{linp} is the existence of $\lambda\in\R^m$ such that:
\begin{enumerate}
\item[(a)] $\mathbf Ax=\mathbf b, x>0$,
\item[(b)] $\nabla f(x)+\mathbf A^\top\lambda\geq-\varepsilon$,
\item[(c)] $\|X(\nabla f(x)+\mathbf A^\top\lambda)\|_\infty\leq\varepsilon$,
\item[(d)] $d^\top(X\nabla^2 f(x)X+\varepsilon\mathcal{I})d\geq0$ for all $d$ such that $AXd=0$. 
\end{enumerate}
\end{proposition}
\begin{proof} The claimed satisfaction of (i)-(iii) in Definition \ref{epskkt2} follow immediately from Proposition \ref{epsilonKKT}.  The following shows (iv) and (v). For all $\varepsilon'>0$ it holds that $d^\top(X\nabla^2 f(x)X+(\varepsilon+\varepsilon')\mathcal{I})d>0$ for all $d\neq0$ such that $AXd=0$. It is well know that, in this case, there is some $\rho>0$ such that $X\nabla^2 f(x)X+(\varepsilon+\varepsilon')\mathcal{I}+\rho X\mathbf A^\top \mathbf AX$ is positive definite (see, for instance, \cite[Proposition 2.1]{cakkt2}). Since $X^{-1}$ is positive definite, we have $\nabla^2 f(x)+\sum_{i=1}^m\frac{\varepsilon+\varepsilon'}{x_i^2}e_ie_i^\mathtt{T}+\rho \mathbf A^\top \mathbf A$ is positive definite, where $e_i$ is the $i$-th canonical vector. Taking the limit $\varepsilon'\to0^+$ and restricting to $d$ with $\mathbf Ad=0$ we have $d^\top(\nabla^2 f(x)+\sum_{i=1}^m\frac{\varepsilon}{x_i^2}e_ie_i^\mathtt{T})d\geq0$ for all $d$ with $\mathbf Ad=0$ and the result follows defining $\theta_i:=\frac{\varepsilon}{x_i^2}, i=1,\dots,n$.\end{proof}

%
%
%
%

\section{Interior Trust-Region Point Algorithms and Computational Complexity for $\varepsilon$-Perturbed Stationary Points}\label{Sec: Complexity}

We once again focus on \eqref{linp} and present two interior trust-region point (ITRP) algorithms  that are theoretically ensured to generate $\varepsilon$-perturbed stationary points. Both algorithms belong to the class of fully polynomial time approximation schemes. Let $\Omega:=\{x\mid \mathbf Ax=\mathbf b, x\geq0\}$ denote the feasible set and $\Omega^\circ:=\{x\mid \mathbf Ax=\mathbf b, x>0\}$ its interior. Assume that  
the feasible region is bounded and has a non-empty interior.
For any given positive $\mu\le 1$, we consider the potential function
\begin{align}\phi(x):=f(x)-\mu\sum_{i=1}^n\log(x_i).\label{define PF}
\end{align}
Note that the gradient of the potential function at $x>0$ is
\[\nabla\phi(x)=\nabla f(x) - \mu X^{-1}e.\]
Then the ITRP algorithms are summarized in Algorithm \ref{Algorithm 1}, where we have a specific initialization rule; we elect to initialize the algorithm with an approximate analytic center $x^0\in\Omega^{\circ}$ that satisfies
\begin{align}
-\sum_{i=1}^n\log(x_i)\ge -\sum_{i=1}^n\log(x_i^0)-O(1),
\end{align}
for all $x:=(x_i)\in\Omega^{\circ}$ for some problem-independent constant $O(1)$. Such an initial solution is efficiently computable.

Meanwhile, we choose to terminate the algorithm when the per-iteration improvement on the potential function is smaller than a certain threshold to be specified soon afterwards. Constants $\mu$ and $\beta$ will also be defined later on.

\begin{algorithm}
\caption{Pseudo-code of the interior trust-region point  (ITRP) algorithm}\label{Algorithm 1}
\begin{description}
\item[Step 1.]Given $\varepsilon\in(0,\,1]$ and choose $x^0\in\Omega^{\circ}$ to be an approximate analytic center of the feasible region. Let  $t:=0$.
\item[Step 2.] Solve the following problem
\begin{align}
\min&~ \begin{cases} 
\nabla\phi(x^t)^\top X_td&\text{first-order ITRP}
\\\nabla\phi(x^t)^\top X_td+\frac{1}{2}d^\top X_t\nabla^2f(x^t)X_td & \text{second-order ITRP}
\end{cases}
\label{sub 1}
\\s.t.&~AX_td=0,\ \|d\|\le \beta;\label{sub 2}
\end{align}
where $X_t=diag(x^t)$.
Denote by $d^t$ the solution.
\item[Step 3.] Update $x^{t+1}:=x^t+X_td^t$.
\item[Step 4.] Algorithm terminates if stopping criterion is satisfied. Otherwise, let $t:=t+1$ and go to Step 2.
\end{description}
\end{algorithm}

In Algorithm \ref{Algorithm 1},  the per-iteration subproblem \eqref{sub 1}-\eqref{sub 2} can be chosen from the first-order or the second-order mode depending on the target of the optimization, that is, to achieve an $\varepsilon$-perturbed first- or second-order stationary point, respectively. Also, the second-order mode yields a perturbed first-order stationary point at a faster complexity rate. In both modes, the resulting per-iteration problem \eqref{sub 1}-\eqref{sub 2}  are easily solvable. Specifically, in the case of first-order ITRP, Problem  \eqref{sub 1}-\eqref{sub 2} admits a closed form solution that does not involve any Hessian information, nor matrix inversion. Therefore, in this case the ITRP belongs to the class of first-order algorithms. In contrast, in the second-order ITRP, the subproblem can be solved using a bisection scheme as per \cite{Ye98,ye1992affine} with a ``log-log'' (quadratic) rate of complexity. 

In the following, we will show that both modes of the ITRP entails the best rate of worst-case iteration complexity known for a stricter class of nonlinear optimization problems. We will make use of the following lemma, which is well known in the literature of interior-point algorithms (e.g., \cite{Karmarkar84}):
\begin{lemma}\label{lemma1}
Let $x>0$ and $\|X^{-1}d\|\le \beta <1$. Then
\[-\sum_{i=1}^n\ln(x_i+d_i)+\sum_{i=1}^n\ln(x_i)\le -e^\top X^{-1}d+\frac{\beta^2}{2(1-\beta)}.\]
\end{lemma}

\subsection{Complexity Analysis for the First-Order ITRP Algorithm}
This subsection presents the complexity analysis for the first-order ITRP with a general assumption that  $f$ is potentially not (directionally) differentiable. In the following, we first present our assumptions in Section \ref{first order assumptions}. Section \ref{complexity analysis first-order} then presents the promised complexity analyses.

\subsubsection{Assumptions for the first-order ITRP}\label{first order assumptions}
Our complexity analysis herein relies on the following set of assumptions.

\bigskip

{\bf\noindent Assumption 3:}
\begin{enumerate}
\item[(a)] Function $f(x)$ is differentiable for all $x\in\Omega^\circ$. In addition, there exists $\gamma\geq1$ such that for all $x\in\Omega^\circ$ and $d\in \{ d: \Vert d\Vert\leq r,\, X(e+d)\in\Omega\}$ for some $r<1$, 
$$f\left(X(e+d)\right)\leq f(x)+\langle X\nabla f(x),d\rangle +\frac{\gamma}{2}\|d\|^2.$$
\item[(b)]  {The feasible region is bounded with $\max\{\|x\|_\infty: x\in\Omega\}\leq R$, for some $R\geq1$}.\\
\item[(c)]The objective function is bounded from below in the feasible set, that is, there exists $L\in\R$ with $f(x)\geq L$ for all $x\in\Omega^\circ$.
\end{enumerate}

\begin{remark}
Assumption 3.(a) subsumes the following special but important cases:
\begin{enumerate}
\item For all $x, x^+\in\Omega$, it holds that $f\left(x^+\right)\leq f(x)+\langle \nabla f(x),x^+-x\rangle +\frac{\hat\beta}{2}\|x^+-x\|^2$ for some $\hat \beta>0$. Such an inequality implies Assumption 3.(a) with $\gamma:=\hat \beta R^2$.
\item Function $f:=f_1+f_2$ is a composite function, with $f_1$ being continuously differentiable and $f_2(x):=\sum_{i=1}^nx_i^p$ for any $p: 0<p<1$. To see this, we may observe that $f_2\left( X(d+e)\right)=\sum_{i=1}^n x_i^p (d_i+1)^p$ for any $d=(d_i)\in\R^n$ and any $x=(x_i)\in\Omega$. Also, $f_2(X(d+e))$ is continuously differentiable in $d$ and the largest eigenvalue of its Hessian in $d$ is upper bounded by $\frac{R^pp(p-1)}{(1-\beta)^{2-p}}$. It is worth noticing that $f_2$ is not differentiable when $x_i=0$ for any $i$.
\end{enumerate}
\end{remark}
\begin{remark}
 {Assumption 3.(b) can be easily generalized to the assumption that the level sets of $f$ are bounded, that is, given $x^0\in\Omega^\circ$, there exists $R\geq1$ such that $\sup\{\|x\|_\infty: f(x)\leq f(x^0), x\in\Omega^\circ\}\leq R$.}
\end{remark}

\subsubsection{Complexity estimate for the first-order ITRP}\label{complexity analysis first-order}

We are now ready to present our complexity analysis. We elect to terminate the algorithm whenever $\phi(x^{t+1})-\phi(x^t)> -\frac{\varepsilon^2}{2\gamma+4\varepsilon}$ and output the solution $x^t$.
\begin{theorem}\label{first order proof}
Suppose that Assumption 3 holds. Denote by $f^*$ the global minimal value of the objective function $f$ on $\Omega$. Consider Algorithm \ref{Algorithm 1} with first-order ITRP per-iteration problem. For any $\varepsilon\in(0,\,\min\{r,\,1\}]$,  let $\mu:=\varepsilon$,  $\beta:=\left(\gamma+2\mu\right)^{-1}\mu$, and $t^*:=\left\lceil\frac{\left(f(x^0)-f^*+O(1)-\varepsilon\right)\left(2\gamma+4\varepsilon\right)}{\varepsilon^2}\right\rceil$, the algorithm terminates before the $t^*$-th iteration at a $2\varepsilon$-KKT point, more precisely, at a feasible solution $\hat x$ that satisfies $\nabla f(\hat x) +{\mathbf A}^\top \hat y>0$ and $\Vert diag(\hat x)\left(\nabla f(\hat x) +{\mathbf A}^\top \hat y\right)\Vert_\infty\leq 2\varepsilon$ for some $\hat y$. Otherwise, it holds that $f(x^{t^*})-f^*\leq\varepsilon.$
\end{theorem}

\begin{proof}

{\bf Step 1.} In this step, we would like to show that $x^t\in\Omega^\circ$ for all $t\geq 1$.  To this end, we notice that, if $x^{t-1}\in\Omega^\circ$, it holds that $x^{t}_i=x_i^{t-1}+x_i^{t-1}d^{t-1}_i=x_i^{t-1}(1+d^{t-1}_i)>0$ for any $i=1,...,n$, where the last inequality is because $\Vert d^{t-1}\Vert\leq \beta<1$ imposed as a constraint in \eqref{sub 2}. Also, if $x^{t-1}\in\Omega^\circ$, it holds that $\mathbf A x^{t}=\mathbf A(x^{t-1}+X_{t-1}d^{t-1})=\mathbf b+\mathbf AX_{t-1}d^{t-1}=\mathbf b$, where the last identity is based on constraint \eqref{sub 2}. Our proof for Step 1 completes by noticing that $x^0\in\Omega^\circ$. 

{\bf Step 2.} In this step, we would like to show that either of the following holds at iteration $k$:
\begin{align}
\phi(x^{t+1})-\phi(x^t)\leq -\frac{\varepsilon^2}{2\gamma+4\varepsilon},\label{whenever holds}
\end{align} 
or $\Vert X_t\nabla f(x^t) - \mu e+X_tA^\top y^t\Vert_\infty<2\varepsilon$ and $\nabla f(x^t) +{\mathbf A}^\top y^t>0$ for some $y^t\in\R^m$.

To this end, we first notice that subproblem \eqref{sub 1}-\eqref{sub 2} can be solved globally, whose first-order optimality condition yields that
\begin{align}
X_t\nabla f(x^t) - \mu e+X_tA^\top y^t+\lambda^t d^t=0,\label{KKT 1}
\end{align}
 for some Lagrange multipliers $y^t\in\R^m$ and $\lambda^t\in\R$.
From the inequality in Assumption 3.(a), since $x^t\in \Omega^\circ$ and $d^t:\, \Vert d^t\Vert\leq \beta=(\gamma+2\mu)^{-1}\mu<\varepsilon\leq 1$ from the result in Step 1, it holds that 
\begin{align}
f\left(X_{t}(e+d^t)\right)\leq f(x^t)+\langle X_t\nabla f(x^t),d^t\rangle +\frac{\gamma}{2}\|d^t\|^2.
\end{align}
Combined with Lemma \ref{lemma1}, it  implies that
\begin{align}
\phi(x^{t+1})-\phi(x^t)\leq& \langle \nabla f (x^t),\, X_td^t\rangle+\frac{\gamma}{2}\Vert d^t\Vert^2-\mu e^\top X_t^{-1}d+\mu\beta^2
\\=& \langle \nabla \phi(x^t),\, X_td^t\rangle+\frac{\gamma}{2}\Vert d^t\Vert^2+\mu\beta^2.
\end{align}
Thus, 
\begin{equation}
\phi(x^{t+1})-\phi(x^t)\leq \langle X_tA^\top y^t-\lambda d^t,\, d^t\rangle+\frac{\gamma}{2}\Vert d^t\Vert^2+\mu\beta^2
=\langle -\lambda^t d^t,\, d^t\rangle+\frac{\gamma}{2}\Vert d^t\Vert^2+\mu\beta^2.\label{inequality aa}
\end{equation}

{\bf Case 1:} If $\Vert d^t\Vert<\beta$, then $\lambda^t=0$ and $X_t\nabla f(x^t) +X_tA^\top y^t=\mu e$. Since $\mu:=\varepsilon>0$, it therefore holds that $\nabla f(x^t) +{\mathbf A}^\top y^t>0$ and that $\Vert X_t\nabla f(x^t) +X_tA^\top y^t\Vert_\infty\leq \varepsilon$.

{\bf Case 2:} Consider the case where $\Vert d^t\Vert=\beta$. Let $p(x,y):=X\nabla f(x) - \mu e+XA^\top y$. (Again, $X:=diag(x)$.)  From \eqref{KKT 1}, it therefore holds that $\Vert p(x^t,y^t)\Vert = \lambda^t\Vert d^t\Vert=\lambda^t\beta$. Combined with \eqref{inequality aa}, it yields that
\begin{align}
\phi(x^{t+1})-\phi(x^t)\leq -\lambda^t\beta^2+\frac{\gamma}{2}\Vert d^t\Vert^2+\mu\beta^2=-\beta\Vert p(x^t, y^t)\Vert +\left(\frac{\gamma}{2}+\mu\right)\beta^2.
\end{align}
{\bf\indent\indent Case 2.1:} Under Case 2, if $\Vert p(x^t,y^t)\Vert \geq \mu$, then 
\begin{align}
\phi(x^{t+1})-\phi(x^t)\leq -\beta\mu +\left(\frac{\gamma}{2}+\mu\right)\beta^2.
\end{align}
Since $\mu:=\varepsilon$ and $\beta:=\left(\gamma+2\mu\right)^{-1}\mu$, we have that
\begin{align}
\phi(x^{t+1})-\phi(x^t)\leq -\frac{\varepsilon^2}{2\gamma+4\varepsilon}.\label{whenever holds}
\end{align}

{\bf\indent\indent Case 2.2:} Under Case 2, if $\Vert p(x^t,y^t)\Vert <\mu$, then
\begin{align}
\Vert X_t\nabla f(x^t) - \mu e+X_tA^\top y^t\Vert_\infty\leq \Vert X_t\nabla f(x^t) - \mu e+X_tA^\top y^t\Vert<\mu,
\end{align}
therefore, $X_t\nabla f(x^t) +X_tA^\top y^t>0\Longrightarrow \nabla f(x^t) +{\mathbf A}^\top y^t>0$. Meanwhile, $\Vert X_t\nabla f(x^t) +X_tA^\top y^t\Vert_\infty<2\mu=2\varepsilon$ for given $\mu:=\varepsilon$. Summarizing the above cases, we know that Case 1, Case 2.1, and Case 2.2 are mutually exclusive. Thus we have the desired result in Step 2.

{\bf Step 3.} We would like to summarize the above steps to obtain the claimed results in this theorem. We first observe that, because the elected initial solution $x^0$ satisfies that
\[-\sum_{i=1}^n\log(x_i^t)\ge -\sum_{i=1}^n\log(x_i^0)-O(1),\]
we have that, if \eqref{whenever holds} holds for all $t\leq t'$, it holds that
\begin{align}
f(x^{t'})-f(x^0)\leq -\frac{t'\varepsilon^2}{2\gamma+4\varepsilon}+O(1).
\end{align}
It therefore holds that $f(x^{t'})-f^*\leq \left[f(x^0)-f^*\right]-\frac{t'\varepsilon^2}{2\gamma+4\varepsilon}+O(1)$.

Recall that the algorithm  terminates whenever $\phi(x^{t+1})-\phi(x^t)> -\frac{\varepsilon^2}{2\gamma+4\varepsilon}$ for some $t$.
Therefore, at iteration $t^*=\frac{\left(f(x^0)-f^*+O(1)-\varepsilon\right)\left(2\gamma+4\varepsilon\right)}{\varepsilon^2}$, it holds either that the algorithm has terminated before iteration $k^*$ at a feasible solution $\hat x$ that satisfies that $\nabla f(\hat x) +{\mathbf A}^\top \hat y>0$ and $\Vert diag(\hat x)\nabla f(\hat x) +\hat XA^\top \hat y\Vert_\infty\leq \varepsilon$. Otherwise, it holds that $f(x^{k^*})-f^*\leq\varepsilon.$
\end{proof}

\begin{remark}
The first-order ITRP solves a constrained problem with potential non-differentiability at an iteration complexity of $O(1/\varepsilon^2)$. For this types of problems, such a rate is best known to the literature. It is also worth emphasizing that the per-iteration problem admits a closed-form solution.
\end{remark}

\subsection{Complexity Analysis for the Second-Order ITRP Algorithm}

This subsection presents the complexity analysis for the second-order ITRP with three different sets of regularities on $f$: (i) $f$ is potentially not twice differentiable; (ii) $f$ is potentially not differentiable; and (iii) $f$ is a quadratic function.  The resulting complexity estimates as well as the characteristics of the final solution output from the ITRP vary according to the changes of assumptions. In the following, we first present our assumptions in Section \ref{second order assumptions}. Section \ref{complexity analysis second-order} then presents the promised complexity analyses.

\subsubsection{Assumptions for the second-order ITRP}\label{second order assumptions}
The analysis on the second-order ITRP relies on the following assumptions.

{\bf \noindent Assumption 4:}
Function $f(x)$ is twice differentiable for all $x\in\Omega^\circ$. For all $x\in\Omega^\circ$ and $d,d'\in\{d: \Vert d\Vert\leq r,\, X(e+d)\in\Omega^\circ\}$, for some $r<1$ and ${\eta}\geq 1$, it holds that
\begin{align}
&\Vert X\nabla^2 f\left(X(e+d)\right)-X\nabla^2 f\left(X(e+d')\right)\Vert\leq {\eta}\Vert d-d'\Vert;\quad\text{and}\nonumber\\
&\nabla f\left(X(e+d)\right)-\nabla f(x)\leq \langle  X \nabla f(x),\,d\rangle
+\frac{1}{2} d^\top X\nabla^2 f\left(x\right)X d+\frac{{\eta}}{3}\Vert d\Vert^3.\label{taylor expansion 1}
\end{align}

{\bf \noindent Assumption 5:}
Function $f(x)$ is twice differentiable for all $x\in\Omega^\circ$. For all $x\in\Omega^\circ$ and $d,d'\in\{d: \Vert d\Vert\leq r,\, X(e+d)\in\Omega^\circ\}$, for some $r<1$ and ${\eta}\geq 1$, it holds that
\begin{align}
&\Vert X\nabla^2 f\left(X(e+d)\right)X-X\nabla^2 f\left(X(e+d')X\right)\Vert\leq {\eta}\Vert d-d'\Vert;\quad\text{and}\nonumber\\
&\nabla f\left(X(e+d)\right)-\nabla f(x)\leq \langle  X \nabla f(x),\,d\rangle
+\frac{1}{2} d^\top X\nabla^2 f\left(x\right)X d+\frac{{\eta}}{3}\Vert d\Vert^3.\label{taylor expansion 2}
\end{align}

\begin{remark}\label{remark 10}
Assumption 4 and Assumption 5 subsume some special but important cases:
\begin{enumerate}
\item For all $x, x^+\in\Omega$, it holds that $f(x)$ is twice differentiable and 
\begin{align}\Vert\nabla^2f(x)-\nabla^2f(x^+)\Vert\leq \hat{\eta}\Vert x-x^+\Vert,\label{twice c2}
\end{align}
 for some ${\hat \eta}>0$. Such an inequality implies  both Assumptions 4 and Assumption 5 with ${\eta}:={\hat \eta} R^3$. These are immediate from the  observation that
 \begin{align} 
  \Vert X\nabla^2f(x)X-X\nabla^2f(x^+)X\Vert\leq \Vert X\Vert^2 {\hat \eta}\Vert x-x^+\Vert\leq\Vert X\Vert^3 \hat\eta\Vert d\Vert,\label{first inequality}
  \\
    \Vert X\nabla^2f(x)-X\nabla^2f(x^+)\Vert\leq \Vert X\Vert {\hat \eta}\Vert x-x^+\Vert\leq\Vert X\Vert^2 \hat\eta\Vert d\Vert,\label{second inequality}
  \end{align}
  as well  as the direct implication of \eqref{twice c2} in the form of
   \begin{align} 
\nabla f\left(X(e+d)\right)-\nabla f(x)\leq &\langle  X \nabla f(x),\,d\rangle
+\frac{1}{2} d^\top X\nabla^2 f\left(x\right)X d+\frac{{\hat\eta}}{3}\Vert X d\Vert^3\nonumber
\\\leq &\langle  X\nabla f(x),\,d\rangle
+\frac{1}{2} d^\top X\nabla^2 f\left(x\right)X d+\frac{R^3{\hat\eta}}{3}\Vert  d\Vert^3. \nonumber
  \end{align}
  
\item Let function $f:=f_1+f_2$ be a composite function, with $f_1$ being twice continuously differentiable. If $f_2(x):=\sum_{i=1}^nx_i^{p}$ for some $p: p>0$ then  for any $d=(d_i)\in\R^n: \Vert d\Vert\leq r<1$, we immediately have  
\begin{align}
\frac{\partial^2 f_2\left( X(d+e)\right)}{ \partial x_i^2}=p(p-1) x_i^{p-2}(d_i+1)^{p-2};\nonumber
\\
x_i\cdot\frac{\partial^2 f_2\left(X(d+e)\right)}{ \partial x_i^2}=p(p-1) x_i^{p-1}(d_i+1)^{p-2};\nonumber
\\(x_i)^2\cdot\frac{\partial^2 f_2\left(X(d+e)\right)}{ \partial x_i^2}=p(p-1) x_i^{p}(d_i+1)^{p-2}.\nonumber
\end{align}
Then, it is easily verifiable that:
\begin{itemize} 
\item  if $p:\, 1<p<2$,  Assumption 4 holds, but $f(x)$ is not twice differentiable for $x\in\{x_i=0,\,\text{for some $i$}\}$.
\item  if $p:\, 0<p<1$, Assumption 5 holds, but $f(x)$ is not differentiable for $x\in\{x_i=0,\,\text{for some $i$}\}$.
\end{itemize}
\end{enumerate}
\end{remark}

\begin{remark}
Assumption 5 subsumes Assumption 4: It is evident that Assumption 4 implies Assumption 5, while the reverse does not hold telling from the second special case in Remark \ref{remark 10}.
\end{remark}


\bigskip

{\bf \noindent Assumption 6:} $f$ is a quadratic function, that is, $\eta=0$.

\bigskip

\subsubsection{Complexity estimates for the second-order ITRP}\label{complexity analysis second-order}

This section presents the complexity estimates for the second-order ITRP under three different sets of assumptions. Theorem \ref{theorem 3.3} first considers the case when $f$ is potentially not twice differentiable and shows that the desired $\varepsilon$-perturbed first- and second-order stationary point can be achieved with a rate of $O(\varepsilon^{-3/2})$ and $O(\varepsilon^{-3})$, respectively. Then, Theorem \ref{theorem 3.4} generalizes to the case where $f$ is potentially not (directionally) differentiable and shows that the same set of  efficiency rates can be achieved in generating a weaker version of the  $\varepsilon$-perturbed first- and second-order stationary point. Such a version of approximate necessary conditions is also studied by \cite{bian}. Finally, Theorem \ref{theorem 3.5} presents a special case where $f$ is a quadratic function. In such a case, the second-order ITRP is especially efficient and achieves the $\varepsilon$-perturbed first- and second-order stationary point both at rate of $O(\varepsilon^{-1})$. Theorem \ref{theorem 3.5} presents an alternative proof for the same result presented in \cite{Ye98}. We should note that the termination criteria for the above three cases are slightly different.

For our first case, we consider the algorithm under  Assumptions 4 and 6. We elect to terminate the second-order ITRP whenever the following criteria hold:
\begin{align}
\phi(x^{t+1})-\phi(x^{t}) > -\frac{\sqrt{\varepsilon^3}}{200{\eta}^2R^{3/2}},\nonumber
\\
\phi(x^{t+2})-\phi(x^{t+1}) > -\frac{\sqrt{\varepsilon^3}}{200{\eta}^2R^{3/2}}.\nonumber
\end{align}
At termination, the algorithm outputs solution $x^{t+1}$.

\begin{theorem}\label{theorem 3.3}
Suppose that Assumptions 3.(b), 3.(c) and 4 hold. Denote by $f^*$ the global minimal value of the objective function $f$ on $\Omega$. Consider Algorithm \ref{Algorithm 1} with second-order ITRP per-iteration problem. For any $\varepsilon\in\left(0,\,\min\left\{10{\eta}^2 r^2,\,\frac{1}{2}\right\}\right]$,  let $\mu:=\frac{\varepsilon}{5{\eta} R}$,  $\beta:=\mu^{1/2}{\eta}^{-1/2}/\sqrt{2}$, and $t^*:=\left\lceil\frac{400{\eta}^2R^{3/2}\left(f(x^0)-f^*+O(1)-\varepsilon\right)\left(2{\eta}+4\varepsilon\right)}{\sqrt{\varepsilon^3}}+1\right\rceil$. The algorithm  terminates before the $t^*$-th iteration at an $\varepsilon$-KKT and $\sqrt\varepsilon$-KKT2 point, more precisely, at a feasible solution $\hat x$ that satisfies, for some $\hat y\in\R^m$, that 
\begin{align}
&\hat x>0,\quad\nabla f(\hat x) +{\mathbf A}^\top \hat y>-\varepsilon,\nonumber
\\&\Vert diag(\hat x)(\nabla f(\hat x) +{\mathbf A}^\top \hat y)\Vert_\infty\leq \varepsilon,\nonumber
\\
&d^\top\left(diag(\hat x)\nabla^2f(\hat x)diag(\hat x)+\sqrt{\varepsilon} I\right)d\geq 0,\quad\forall d:\, \mathbf A diag(\hat x) d=0.\nonumber
\end{align}
 Otherwise, it holds that $f(x^{t^*})-f^*\leq\varepsilon.$
\end{theorem}
\begin{proof} 
{\bf Step 1.} Following Step 1 of the proof for Theorem \ref{first order proof}, it is straightforward that $x^t\in\Omega^\circ$ for all $t\geq 1$. 

{\bf Step 2.} We would like to show that if $\phi(x^{t+1})-\phi(x^t) > -\frac{\sqrt{2{\eta}\mu^3}}{24{\eta}}$ then
$\nabla^2f(x^t)X_td^t-\mathbf A^\top y^t+ \nabla f(x^t)>0$
and
$0\le x_i(\nabla f(x^t)+\nabla^2f(x^t)d^t-A^\top y^t)_i\le 2\mu,\ \forall i,$ for $\beta:=\mu^{1/2}{\eta}^{-1/2}/\sqrt{2}$ and some $y^t\in\R^m$.

To this end, combine Assumption 4 with both $\Vert d^t\Vert\leq \beta=\mu^{1/2}{\eta}^{-1/2}/\sqrt{2}\leq r$ and  Lemma \ref{lemma1}. It therefore holds that
\begin{align}
&~~~\phi(x^{t+1})-\phi(x^t)\nonumber
\\&\le \nabla f(x^t)^\top X_td^t+\frac{1}{2}(d^t)^\top X_t\nabla^2f(x^t)X_td^t+\frac{{\eta}}{3}\|d^t\|^3-\mu e^\top X_t^{-1}d^t+\mu\beta^2\nonumber\\                       
                       &= \nabla\phi(x^t)^\top X_td^t+\frac{1}{2}(d^t)^\top X_t\nabla^2f(x^t) X_td^t+\frac{{\eta}}{3}\|d^t\|^3+\mu\beta^2\nonumber\\
                       &\le \nabla\phi(x^t)^\top X_td^t+\frac{1}{2}(d^t)^\top X_t\nabla^2f(x^t)X_td^t+\left(\frac{{\eta}}{3}\beta+\mu\right)\beta^2.\label{descent inequality here}
                       \end{align}
    Then, the necessary and sufficient global optimality conditions of the trust-region subproblem, besides the feasibility of $d^t$, are
\begin{align}
\begin{split}
&(X_t\nabla^2f(x^t)X_t+\lambda^t I)d^t-X_t\mathbf A^\top y^t=-X_t \nabla\phi(x^t);
\\
&(X_t\nabla^2f(x^t)X_t+\lambda^t I)_{AX_t}\succeq 0,\quad \lambda^t\ge 0,\quad\lambda^t(\beta-\|d^t\|)=0;
\end{split}
\label{second order condition per-iteration}
\end{align}
for Lagrange multipliers $y^t\in\R^m$ and $\lambda^t\in\R$, see \cite{vavasis,sorensen,gay}. Here, $(X_t\nabla^2f(x^t)X_t+\lambda^t I)_{AX_t}\succeq 0$ means 
\[d^\top (X_t\nabla^2f(x^t)X_t+\lambda^t I)d\ge 0,\ \forall d\in \{d:\ AX_td=0\}.\]

If $\|d'\|=\beta$, 
let vector
\[p(x^t,y^t)= X_t\nabla^2f(x^t)X_td^t-X_t\mathbf A^\top y^t+X_t \nabla\phi(x^t).\]
Then from \eqref{second order condition per-iteration}, we have
\begin{equation}\label{lambda}
\lambda^t d^t=-p(x^t,y^t).\end{equation}
Thus,
\begin{align}
&\nabla\phi(x^t)^\top X_t  d^t+\frac{1}{2}(d^t)^\top  X_t\nabla^2f(x^t)X_td^t\nonumber
\\=&\frac{1}{2}\nabla\phi(x^t)^\top X_t d^t+\frac{1}{2}(d^t)^\top (X_t\nabla\phi(x^t)+X_t\nabla^2f(x^t)X_td^t)\nonumber\\
= &\frac{1}{2}(\nabla\phi(x^t)^\top X_t -A^\top y^t)^\top d^t+\frac{1}{2}(d^t)^\top (X_t\nabla\phi(x^t)+X_t\nabla^2f(x^t)X_td^t-A^\top y)\nonumber\\
=& -\frac{1}{2}(d^t)^\top (X_t\nabla^2f(x^t)X_t+\lambda^t I)d^t+\frac{1}{2}(d^t)^\top p(x^t,y^t)\nonumber\\
\le &\frac{1}{2}(d^t)^\top p(x^t,y^t)= -\frac{1}{2}\lambda^t\|d^t\|^2,\label{kkt used here now}
\end{align}
where \eqref{kkt used here now} is immediately due to \eqref{lambda}. 

As an immediate result, combined with \eqref{descent inequality here}, it holds that
\begin{align}
&\phi(x^{t+1})-\phi(x^t)\le -\frac{1}{2}\lambda^t\|d^t\|^2+\left(\frac{{\eta}}{3}\beta+\mu\right)\beta^2
\\=&-\frac{1}{2}\lambda^t\|d^t\|^2+\left(\frac{{\eta}}{3}\mu^{1/2}{\eta}^{-1/2}/\sqrt{2}+\mu\right)\mu{\eta}^{-1}/2
\\=&-\frac{1}{2}\lambda^t\|d^t\|^2+\left(\frac{\sqrt{2{\eta}\mu^3}}{12{\eta}}+\frac{\mu^2}{2{\eta}}\right).
\end{align}
Recall that ${\eta}\geq 1$ and $\varepsilon\leq \frac{1}{2}\leq \frac{5{\eta} R^2}{2}\Longrightarrow \mu\leq \frac{{\eta}}{8}\Longrightarrow \frac{\sqrt{2{\eta}\mu^3}}{12{\eta}}+\frac{\mu^2}{2{\eta}}\leq \frac{5\sqrt{2{\eta}\mu^3}}{24{\eta}}$.
If $\phi(x^{t+1})-\phi(x^t) > -\frac{\sqrt{2{\eta}\mu^3}}{24{\eta}}$, then $-\frac{\sqrt{2{\eta}\mu^3}}{24{\eta}}< -\frac{1}{2}\lambda^t\|d^t\|^2+\left(\frac{{\eta}}{3}\beta+\mu\right)\beta^2\Longrightarrow \frac{1}{2}\lambda^t\|d^t\|^2 < \frac{\sqrt{2{\eta}\mu^3}}{4{\eta}}$. We might consider the following two cases.

{\bf Case 1.} If  $\Vert d^t\Vert<\beta$, it then holds that $\lambda^t=0$. As a result, condition \eqref{second order condition per-iteration} yields that
\begin{align}
X_t\nabla^2f(x^t)X_td^t-X_t\mathbf A^\top y^t+X_t \nabla\phi(x^t)=0;\quad (X_t\nabla^2f(x^t)X_t)_{AX_t}\succeq 0.\label{case 1 result 2nd}
\end{align}
Thus, it holds that
\begin{align}
\Vert X_t\nabla^2f(x^t)X_td^t-X_t\mathbf A^\top y^t+X_t\nabla f(x^t)\Vert_\infty = \mu< 2\mu,
\end{align}
and 
\begin{align}
\nabla^2f(x^t)X_td^t-\mathbf A^\top y^t+\nabla f(x^t)>0.
\end{align}


{\bf\indent Case 2.} If $\Vert d^t\Vert=\beta$, then $\|p(x^t,y^t)\|=\lambda^t\beta$. Thus
\begin{align}\frac{\sqrt{2{\eta}\mu^3}}{4{\eta}}&> \frac{1}{2}\lambda^t\|d^t\|^2=\frac{1}{2}\lambda^t\beta^2 = \frac{1}{2}\beta\|p(x^t,y^t)\|=\frac{\sqrt{2{\eta} \mu}}{4}\|p(x^t,y^t)\|,\nonumber
\end{align}
which means that $\|p(x^t,y^t)\|< \mu$, that is,
\begin{align}\mu>& \|X_t\nabla^2f(x^t)X_td^t-X_t\mathbf A^\top y^t+X_t \nabla\phi(x^t)\|_{\infty} \nonumber
\\=&\|(X_t\nabla^2f(x^t)X_td^t-X_t\mathbf A^\top y^t+X_t \nabla f(x^t)) -\mu e\|_{\infty},\nonumber
\end{align}
which implies
\[\nabla^2f(x^t)X_td^t-\mathbf A^\top y^t+ \nabla f(x^t)>0,\]
and
\[0\le x_i(\nabla f(x^t)+\nabla^2f(x^t)d^t-A^\top y^t)_i\le 2\mu,\ \forall i.\]

Combining Cases 1 and  2, we have the desired result in Step 2.

{\bf Step 3.} We would like to show that once it holds that 
\begin{align}&\nabla^2f(x^t)X_td^t-\mathbf A^\top y^t+ \nabla f(x^t)>0; \nonumber
\\ \text{and} \quad& 0\le x_i(\nabla f(x^t)+\nabla^2f(x^t)d^t-A^\top y^t)_i\le 2\mu,\ \forall i.\label{t indicator 1}
\end{align}
then, it simultaneously holds that, for some $\hat y\in\R^m$:
\begin{align}
\begin{split}
\nabla f(x^{t+1})-\mathbf A^\top \hat y>-\frac{\mu }{2}
\\
\vert x_i^{t+1}(\nabla f(x^{t+1})-A^\top \hat y)_i\vert \leq 4\mu+\mu R,~~\forall i.
\end{split}\label{first order kkt approximate}
\end{align}

To that end, notice that, since $x^t,\,x^{t+1}\in\Omega^\circ$, from mean value theorem, it holds that, for some $\tau\in[0,\,1]$,
\begin{align}
&\nabla f(x^{t+1})-\nabla f(x^t)\nonumber
\\=&\nabla^2 f(\tau(x^{t+1}-x^t)+x^t)(x^{t+1}-x^{t})=\nabla^2 f(\tau(x^{t+1}-x^t)+x^t)X_td^t,
\end{align}
and thus
\begin{align}
&\Vert\nabla f(x^{t+1})-\nabla f(x^t)-\nabla^2 f(x^t) X_td^t\Vert\nonumber
\\=&\Vert\left(\nabla^2 f(x^t)- \nabla^2 f(\tau(x^{t+1}-x^t)+x^t)\right)X_td^t\Vert\nonumber
\\=&\Vert\left(\nabla^2 f(X_t e)- \nabla^2 f(X_t(\tau d^t+e))\right)X_t\Vert \Vert d^t\Vert\nonumber
\\\leq& {\eta} \tau\Vert d^t\Vert^2\leq  {\eta} \Vert d^t\Vert^2\leq {\eta} \beta^2,\label{mean value theorem}
\end{align}
where the last line is due to Assumption 4.(a), combined with $\Vert d\Vert\leq \beta<r$ and $x^t,\, x^{t+1}\in\Omega^\circ$,
which will be useful soon afterwards.

Similarly, we also have
\begin{align}
&\Vert X_t\nabla f(x^{t+1})-X_t\nabla f(x^t)-X_t\nabla^2 f(x^t) X_td^t\Vert\nonumber
\\=&\Vert X_t \left(\nabla^2 f(x^t)- \nabla^2 f(\tau(x^{t+1}-x^t)+x^t)\right)X_td^t\Vert\nonumber
\\\leq& {\eta} \Vert X_t\Vert\Vert d^t\Vert^2\leq  {\eta} R\Vert d^t\Vert^2\leq {\eta} R\beta^2,\label{mean value theorem 2}
\end{align}

Combining \eqref{t indicator 1} with \eqref{mean value theorem}, we have that
\begin{align}
&\nabla f(x^{t+1})-\mathbf A^\top y^t\nonumber
\\\geq  &\nabla^2f(x^t)X_td^t-\mathbf A^\top y^t+ \nabla f(x^t)-\Vert \nabla f(x^{t+1})-\nabla f(x^t)-\nabla^2 f(x^t) X_td^t\Vert_\infty\nonumber
\\\geq& -{\eta} \beta^2=-\frac{\mu}{2}.\nonumber
\end{align}

Meanwhile, combining  \eqref{t indicator 1} with \eqref{mean value theorem 2}, it obtains that
\begin{align}
&\vert x_i^{t+1}(\nabla f(x^{t+1})-A^\top y^t)_i\vert \nonumber
\\\leq & \vert (1+d_i^t)x_i^t(\nabla f(x^t)+\nabla^2f(x^t)d^t-A^\top y^t)_i\vert\nonumber
\\&+\vert 1+d_i^t\vert\cdot\Vert X_t\nabla f(x^{t+1})-X_t\nabla f(x^t)-X_t\nabla^2 f(x^t) X_td^t\Vert_\infty\nonumber
\\\leq &(1+\beta)(2\mu+{\eta} R)\beta^2 \leq (1+\beta)\left(2\mu+\frac{\mu R}{2}\right)\leq 4\mu+\mu R.\nonumber
\end{align}
The last line is due to $\vert 1+d_i^t\vert\leq (1+\beta)\leq 2$.

{\bf Step 4.}  We would like to show that, if $\phi(x^{t+2})-\phi(x^{t+1})>-\frac{\sqrt{2{\eta}\mu^3}}{24{\eta}}$, then $(X_{t+1}\nabla^2f(x^{t+1})X_{t+1}+\sqrt{2\mu{\eta}} I)_{AX_{t+1}}\succeq 0$. To this end, we invoke \eqref{lambda} (where we let $t:=t+1$), \eqref{case 1 result 2nd} (where we let $t:=t+1$), and \eqref{second order condition per-iteration} (where we let $t:=t+1$). The combination of the three results in 
\begin{align}
\left(X_{t+1}\nabla^2f(x^{t+1})X_{t+1}+\frac{\|p(x^{t+1},y^{t+1})\|}{\beta} I\right)_{AX_{t+1}}\succeq 0.\label{second order approximate second}
\end{align}
Further observe that from Step 2, it holds that, if $\phi(x^{t+2})-\phi(x^{t+1})>-\frac{\sqrt{2{\eta}\mu^3}}{24{\eta}}$, then $\frac{\|p(x^{t+1},y^{t+1})\|}{\beta}\leq \frac{\mu}{\beta}=\sqrt{2\mu{\eta}}$. Combined with \eqref{second order approximate second}, we have the claimed result in this step.

{\bf Step 5.} This step summarizes the above steps and prove the claimed results of the theorem.

We recall here $x^0$ is the approximate analytic center that satisfies 
\begin{align}
-\sum_{i=1}^n\log(x_i^t)\ge -\sum_{i=1}^n\log(x_i^0)-O(1),\label{approximate analytic center here}
\end{align}
where $O(1)$ is a constant.

We know that
 at iteration $t^*:=\frac{400{\eta}^2R^{3/2}\left(f(x^0)-f^*+O(1)-\varepsilon\right)\left(2{\eta}+4\varepsilon\right)}{\sqrt{\varepsilon^3}}+1$, where $O(1)$ is the same number as in \eqref{approximate analytic center here} if the termination criteria of simultaneously satisfying
\begin{align}
\phi(x^{t+1})-\phi(x^t) > -\frac{\sqrt{\varepsilon^3}}{200{\eta}^2R^{3/2}}>-\frac{\sqrt{2{\eta}\mu^3}}{24{\eta}}=-\frac{\sqrt{10\varepsilon^3}}{600{\eta}^2R^{3/2}},\nonumber
\\
\phi(x^{t+2})-\phi(x^{t+1}) > -\frac{\sqrt{\varepsilon^3}}{200{\eta}^2R^{3/2}},\nonumber
\end{align}
have never been satisfied, then, we obtain a reduction in the potential function:
\begin{align}
\phi(x^{t*})-\phi(x^0)\leq -\frac{\sqrt{\varepsilon^3}(t^*-1)}{400{\eta}^2R^{3/2}}=-f(x^0)+f^*-O(1)+\varepsilon.
\end{align}
Then combined with  \eqref{approximate analytic center here}, it holds that
\begin{align}
f(x^{t*})-f(x^0)-O(1)\leq& -\frac{\sqrt{\varepsilon^3}(t^*-1)}{400{\eta}^2R^{3/2}}=-f(x^0)+f^*-O(1)+\varepsilon\nonumber
\\&\Longrightarrow f(x^{t*})-f^*\leq \varepsilon.
\end{align}

Otherwise, the algorithm terminates before $t^*$ and achieves a solution that satisfies \begin{align}
\begin{split}
\nabla f(x^{t+1})-\mathbf A^\top \hat y>-\frac{\mu }{2}>-\varepsilon,
\\
\vert x_i^{t+1}(\nabla f(x^{t+1})-A^\top \hat y)_i\vert \leq 4\mu+\mu R\leq \varepsilon,~~\forall i,
\end{split}
\end{align} according to Step 2. Furthermore, from Step 4, the satisfaction of the termination criteria also implies 
\begin{align}&\left(X_{t+1}\nabla^2f(x^{t+1})X_{t+1}+\sqrt{2\mu{\eta}}I\right)_{AX_{t+1}}\succeq 0\nonumber
\\\Longrightarrow&\left(X_{t+1}\nabla^2f(x^{t+1})X_{t+1}+\sqrt{\varepsilon} I\right)_{AX_{t+1}}\succeq 0,\nonumber
\end{align}
thus immediately leads to the desired result.
\end{proof}

Consider the same algorithm procedure as in the second-order ITRP. If the regularity on $f$ is relaxed from Assumption 4 to Assumption 5, then we may still obtain an approximate KKT condition. Nonetheless, such an approximation is in a critically weaker form. Specifically, we have the following theorem. In this case, we have a slightly different termination criterion: we elect to terminate the second-order ITRP whenever the following criteria hold:
\begin{align}
\phi(x^{t+1})-\phi(x^t) > -\frac{\sqrt{\varepsilon^3}}{200{\eta}^2},\nonumber
\\
\phi(x^{t+2})-\phi(x^{t+1}) > -\frac{\sqrt{\varepsilon^3}}{200{\eta}^2}.\nonumber
\end{align}
Once the algorithm terminates, it outputs $x^{t+2}$ as our final solution.

\begin{theorem}\label{theorem 3.4}
Suppose that Assumptions 3.(b) and 3.(c) and 5 hold. Denote by $f^*$ the global minimal value of the objective function $f$ on $\Omega$. Consider Algorithm \ref{Algorithm 1} with second-order ITRP per-iteration problem. For any $\varepsilon\in\left(0,\,\min\left\{10{\eta}^2 r^2,\,\frac{1}{2}\right\}\right]$,  let $\mu:=\frac{\varepsilon}{5{\eta}}$,  $\beta:=\mu^{1/2}{\eta}^{-1/2}/\sqrt{2}$, and $t^*:=\left\lceil\frac{400{\eta}^2\left(f(x^0)-f^*+O(1)-\varepsilon\right)\left(2{\eta}+4\varepsilon\right)}{\sqrt{\varepsilon^3}}+1\right\rceil$. The algorithm  terminates before the $t^*$-th iteration at a feasible solution $\hat x$ that satisfies that 
\begin{align}
\begin{split}
&\hat x>0,\quad\Vert diag(\hat x)(\nabla f(\hat x) +{\mathbf A}^\top \hat y)\Vert_\infty\leq \varepsilon,
\\
&d^\top\left(diag(\hat x)\nabla^2f(\hat x)diag(\hat x)+\sqrt{\varepsilon} I\right)d\geq 0,\quad\forall d:\, \mathbf A diag(\hat x) d=0.
\end{split}\label{KKT condition approximate weaker achieved old}
\end{align}
 Otherwise, it holds that $f(x^{t^*})-f^*\leq\varepsilon.$
\end{theorem}

\begin{proof} 
{\bf Step 1.} Following Step 1 of the proof for Theorem \ref{first order proof}, it is straightforward that $x^t\in\Omega^\circ$ for all $t\geq 1$. 

{\bf Step 2.} Following Step 2 of the proof for Theorem \ref{theorem 3.3}, it is also evident that, if $\phi(x^{t+1})-\phi(x^t) > -\frac{\sqrt{2{\eta}\mu^3}}{24{\eta}}$ then
$0\le x_i(\nabla f(x^t)+\nabla^2f(x^t)d^t-A^\top y^t)_i\le 2\mu,\ \forall i,$ for $\beta:=\mu^{1/2}{\eta}^{-1/2}/\sqrt{2}$.

{\bf Step 3.} We would like to show that once it holds that 
\begin{align} 0\le x_i(\nabla f(x^t)+\nabla^2f(x^t)d^t-A^\top y^t)_i\le 2\mu,\ \forall i.\label{t indicator 1 new}
\end{align}
then, it holds that, for some $\hat y\in\R^m$:
\begin{align}
\vert x_i^{t+1}(\nabla f(x^{t+1})-A^\top \hat y)_i\vert \leq 5\mu,~~\forall i.\label{first order kkt new}
\end{align}

To that end, notice that, since $x^t,\,x^{t+1}\in\Omega^\circ$, from mean value theorem, it holds that, for some $\tau\in[0,\,1]$,
\begin{align}
&\nabla f(x^{t+1})-\nabla f(x^t)=\nabla^2 f(\tau(x^{t+1}-x^t)+x^t)(x^{t+1}-x^{t})\nonumber
\\=&\nabla^2 f(\tau(x^{t+1}-x^t)+x^t)X_td^t,\nonumber
\end{align}
and thus
\begin{align}
&\Vert X_t\nabla f(x^{t+1})-X_t\nabla f(x^t)-X_t\nabla^2 f(x^t) X_td^t\Vert\nonumber
\\=&\Vert X_t\left(\nabla^2 f(x^t)- \nabla^2 f(\tau(x^{t+1}-x^t)+x^t)\right)X_td^t\Vert\nonumber
\\=&\Vert X_t\left(\nabla^2 f(X_t e)- \nabla^2 f(X_t(\tau d^t+e))\right)X_t\Vert \Vert d^t\Vert\nonumber
\\\leq& {\eta} \tau\Vert d^t\Vert^2\leq  {\eta} \Vert d^t\Vert^2\leq {\eta} \beta^2,\label{mean value theorem new}
\end{align}
where the last line is due to Assumption 5, combined with $\Vert d\Vert\leq \beta<r$ and $x^t,\, x^{t+1}\in\Omega^\circ$,
which will be useful soon afterwards.

Combining \eqref{t indicator 1} with \eqref{mean value theorem new}, we have that
\begin{align}
&\nabla f(x^{t+1})-\mathbf A^\top y^t\nonumber
\\
\geq & \nabla^2f(x^t)X_td^t-\mathbf A^\top y^t+ \nabla f(x^t)-\Vert \nabla f(x^{t+1})-\nabla f(x^t)-\nabla^2 f(x^t) X_td^t\Vert_\infty\nonumber
\\\geq &-{\eta} \beta^2=-\frac{\mu}{2}.\nonumber
\end{align}

Meanwhile, combining  \eqref{t indicator 1 new} with \eqref{mean value theorem new}, it obtains that
\begin{align}
&\vert x_i^{t+1}(\nabla f(x^{t+1})-A^\top y^t)_i\vert \nonumber
\\\leq & \vert (1+d_i^t)x_i^t(\nabla f(x^t)+\nabla^2f(x^t)d^t-A^\top y^t)_i\vert\nonumber
\\&+\vert 1+d_i^t\vert\cdot\Vert X_t\nabla f(x^{t+1})-X_t\nabla f(x^t)-X_t\nabla^2 f(x^t) X_td^t\Vert_\infty\nonumber
\\\leq &(1+\beta)(2\mu+{\eta} )\beta^2 \leq (1+\beta)\left(2\mu+\frac{\mu }{2}\right)\leq 5\mu.\nonumber
\end{align}
The last line is due to $\vert 1+d_i^t\vert\leq (1+\beta)\leq 2$.

{\bf Step 4.}  We would like to show that, if $\phi(x^{t+2})-\phi(x^{t+1})>-\frac{\sqrt{2{\eta}\mu^3}}{24{\eta}}$, then $(X_{t+1}\nabla^2f(x^{t+1})X_{t+1}+\sqrt{2\mu{\eta}} I)_{AX_{t+1}}\succeq 0$. To this end, we invoke \eqref{lambda} (where we let $t:=t+1$), \eqref{case 1 result 2nd} (where we let $t:=t+1$), and \eqref{second order condition per-iteration} (where we let $t:=t+1$). The combination of the three results in 
\begin{align}
\left(X_{t+1}\nabla^2f(x^{t+1})X_{t+1}+\frac{\|p(x^{t+1},y^{t+1})\|}{\beta} I\right)_{AX_{t+1}}\succeq 0.\label{second order approximate second}
\end{align}
Further observe that from Step 2, it holds that, if $\phi(x^{t+2})-\phi(x^{t+1})>-\frac{\sqrt{2{\eta}\mu^3}}{24{\eta}}$, then $\frac{\|p(x^{t+1},y^{t+1})\|}{\beta}\leq \frac{\mu}{\beta}=\sqrt{2\mu{\eta}}$. Combined with \eqref{second order approximate second}, we have the claimed result in this step.

{\bf Step 5.} This step summarizes the above steps and prove the claimed results of the theorem.

We recall here $x^0$ is the approximate analytic center that satisfies 
\begin{align}
-\sum_{i=1}^n\log(x_i^t)\ge -\sum_{i=1}^n\log(x_i^0)-O(1),\label{approximate analytic center here}
\end{align}
where $O(1)$ is a constant.

We know that
 at iteration $t^*=\frac{400{\eta}^2\left(f(x^0)-f^*+O(1)-\varepsilon\right)\left(2{\eta}+4\varepsilon\right)}{\sqrt{\varepsilon^3}}+1$, where $O(1)$ is the same number as in \eqref{approximate analytic center here} if the termination criteria of simultaneously satisfying
\begin{align}
\phi(x^{t+1})-\phi(x^t) > -\frac{\sqrt{\varepsilon^3}}{200{\eta}^2}>-\frac{\sqrt{2{\eta}\mu^3}}{24{\eta}}=-\frac{\sqrt{10\varepsilon^3}}{600{\eta}^2},\nonumber
\\
\phi(x^{t+2})-\phi(x^{t+1}) > -\frac{\sqrt{\varepsilon^3}}{200{\eta}^2},\nonumber
\end{align}
have never been satisfied. Then, we obtain a reduction in the potential function:
\begin{align}
\phi(x^{t*})-\phi(x^0)\leq -\frac{\sqrt{\varepsilon^3}(t^*-1)}{400{\eta}^2}=-f(x^0)+f^*-O(1)+\varepsilon.
\end{align}
Then combined with  \eqref{approximate analytic center here}, it holds that
\begin{align}
f(x^{t*})-f(x^0)-O(1)\leq& -\frac{\sqrt{\varepsilon^3}(t^*-1)}{400{\eta}^2}=-f(x^0)+f^*-O(1)+\varepsilon\nonumber
\\\Longrightarrow &f(x^{t*})-f^*\leq \varepsilon.\nonumber
\end{align}

Otherwise, the algorithm terminates before $t^*$ and achieves a solution that satisfies \begin{align}
\vert x_i^{t+1}(\nabla f(x^{t+1})-A^\top \hat y)_i\vert \leq 5\mu\leq \varepsilon,~~\forall i,
\end{align} according to Step 2. Furthermore, from Step 4, the satisfaction of the termination criteria also implies 
\begin{align}\left(X_{t+1}\nabla^2f(x^{t+1})X_{t+1}+\sqrt{2\mu{\eta}}I\right)_{AX_{t+1}}\succeq 0\nonumber
\\
\Longrightarrow\left(X_{t+1}\nabla^2f(x^{t+1})X_{t+1}+\sqrt{\varepsilon} I\right)_{AX_{t+1}}\succeq 0,\nonumber
\end{align}
thus immediately leads to the desired result.
\end{proof}

\begin{remark}
We observe that even though \eqref{KKT condition approximate weaker achieved old} is a weaker condition than the desired one in this paper, it still applies to application problems such as the non-Lipschitz problem formulation of sparse optimization discussed by \cite{bian}, who provide a different algorithm with the same complexity for a special case that satisfies all our assumptions.
\end{remark}

We now consider a special case where  substantially faster iteration complexity can be achieved. Such a result is, in fact, first presented by \cite{Ye98} for achieving an approximate first-order KKT point for linearly constrained nonconvex quadratic program. The complexity in the approximation to the second-order necessary condition has not been explicitly stated, though a closer look at the results therein may find it an immediate result from the paper.  In the following, we provide an alternative proof for the complexity analysis, which results in some new insights in solving this type of problem. We elect to terminate the second-order ITRP whenever the following criteria hold:
\begin{align}
\phi(x^{t+1})-\phi(x^t) > -\frac{\varepsilon}{32},\nonumber
\\
\phi(x^{t+2})-\phi(x^{t+1}) > -\frac{\varepsilon}{32}.\nonumber
\end{align}
Once the algorithm terminates, it outputs $x^{t+2}$ as our final solution.

\begin{theorem}\label{theorem 3.5}
Suppose that Assumptions 3.(b), 3.(c) and 6 hold. Denote by $f^*$ the global minimal value of the objective function $f$ on $\Omega$. Consider Algorithm \ref{Algorithm 1} with second-order ITRP per-iteration problem. For any $\varepsilon\in\left(0,\,\min\left\{10{\eta}^2 r^2,\,\frac{1}{2}\right\}\right]$,  let $\mu:=\frac{\varepsilon}{4}$,  $\beta:=1/4$, and $t^*:=\left\lceil\frac{64(f(x^0)-f^*+O(1)-\varepsilon)+1}{\epsilon}\right\rceil$, the algorithm  terminates before the $t^*$-th iteration at an $\varepsilon$-KKT2 point, more precisely, at a feasible solution $\hat x$ that satisfies that 
\begin{align}
\begin{split}
&\hat x>0,\quad\nabla f(\hat x)-\mathbf A^\top \hat y>0; \quad\Vert diag(\hat x)(\nabla f(\hat x) +{\mathbf A}^\top \hat y)\Vert_\infty\leq \varepsilon,
\\
&d^\top\left(diag(\hat x)\nabla^2f(\hat x)diag(\hat x)+{\varepsilon} I\right)d\geq 0,\quad\forall d:\, \mathbf A diag(\hat x) d=0.
\end{split}\label{KKT condition approximate weaker achieved}
\end{align}
 Otherwise, it holds that $f(x^{t^*})-f^*\leq\varepsilon.$
\end{theorem}
\begin{proof}
{\bf Step 1.} Following Step 1 of the proof for Theorem \ref{first order proof}, it is straightforward that $x^t\in\Omega^\circ$ for all $t\geq 1$. 

{\bf Step 2.}  We would like to show that if $\phi(x^{t+1})-\phi(x^t) > -\frac{\mu}{16}$ then
$0\le x_i(\nabla f(x^t)+\nabla^2f(x^t)d^t-A^\top y^t)_i\le 2\mu,\ \forall i,$ for $\beta:=1/4$.

Following Step 2 of the proof for Theorem \ref{theorem 3.3}, while noticing that $\eta=0$, we can show that it is also evident that,
\begin{align}
\phi(x^{t+1})-\phi(x^t)\le -\frac{1}{2}\lambda^t\|d^t\|^2+\mu\beta^2.
\end{align}

{\bf Case 1.} If  $\Vert d^t\Vert<\beta$, it then holds that $\lambda^t=0$. As a result, condition \eqref{second order condition per-iteration} yields that
\begin{align}
X_t\nabla^2f(x^t)X_td^t-X_t\mathbf A^\top y^t+X_t \nabla\phi(x^t)=0;\quad (X_t\nabla^2f(x^t)X_t)_{AX_t}\succeq 0.\label{case 1 result 2nd}
\end{align}
Thus, it holds that
\begin{align}
\Vert X_t\nabla^2f(x^t)X_td^t-X_t\mathbf A^\top y^t+X_t\nabla f(x^t)\Vert_\infty = \mu< 2\mu,
\end{align}
and 
\begin{align}
\nabla^2f(x^t)X_td^t-\mathbf A^\top y^t+\nabla f(x^t)>0.
\end{align}


{\bf\indent Case 2.} If $\Vert d^t\Vert=\beta$, then $\|p(x^t,y^t)\|=\lambda^t\beta$. Combined with $\mu\beta^2=\frac{\mu}{16}$, it holds that
\begin{align}
\frac{\mu}{8}&> \frac{1}{2}\lambda^t\|d^t\|^2=\frac{1}{2}\lambda^t\beta^2 = \frac{1}{2}\beta\|p(x^t,y^t)\|=\frac{1}{8}\|p(x^t,y^t)\|.\nonumber
\end{align}
which means that $\|p(x^t,y^t)\|< \mu$, that is, 
\begin{align}
\mu>& \|X_t\nabla^2f(x^t)X_td^t-X_t\mathbf A^\top y^t+X_t \nabla\phi(x^t)\|_{\infty} \nonumber
\\=&\|(X_t\nabla^2f(x^t)X_td^t-X_t\mathbf A^\top y^t+X_t \nabla f(x^t)) -\mu e\|_{\infty},\nonumber
\end{align}
which implies
\[\nabla^2f(x^t)X_td^t-\mathbf A^\top y^t+ \nabla f(x^t)>0,\]
and
\[0\le x_i(\nabla f(x^t)+\nabla^2f(x^t)d^t-A^\top y^t)_i\le 2\mu,\ \forall i.\]

Combining Cases 1 and  2, we have the desired result in Step 2.

{\bf Step 3.} We would like to show that once it holds that 
\begin{align}&\nabla^2f(x^t)X_td^t-\mathbf A^\top y^t+ \nabla f(x^t)>0; \nonumber
\\
\text{and} \quad &0\le x_i(\nabla f(x^t)+\nabla^2f(x^t)d^t-A^\top y^t)_i\le 2\mu,\ \forall i,\label{t indicator 1 new 2}
\end{align}
then, it simultaneously holds that, for some $\hat y\in\R^m$:
\begin{align}
\begin{split}
\nabla f(x^{t+1})-\mathbf A^\top \hat y>0,
\\
\vert x_i^{t+1}(\nabla f(x^{t+1})-A^\top \hat y)_i\vert \leq \mu,~~\forall i.
\end{split}\label{first order kkt approximate}
\end{align}

To that end, notice that, due to Assumption 7, 
\begin{align}
\nabla f(x^{t+1})-\nabla f(x^t)=\nabla^2 f(x^t) X_td^t.\label{niube}
\end{align}

Combining \eqref{t indicator 1 new 2} with \eqref{niube}, we have that
\begin{align}
\nabla f(x^{t+1})-\mathbf A^\top y^t&=  \nabla^2f(x^t)X_td^t-\mathbf A^\top y^t+ \nabla f(x^t)>0.\nonumber
\end{align}

Meanwhile, combining  \eqref{t indicator 1 new 2} with \eqref{niube}, it obtains that
\begin{align}
&\vert x_i^{t+1}(\nabla f(x^{t+1})-A^\top y^t)_i\vert \nonumber
\\\leq & \vert (1+d_i^t)x_i^t(\nabla f(x^t)+\nabla^2f(x^t)d^t-A^\top y^t)_i\vert\nonumber
\\\leq &2\mu(1+\beta)\beta^2\leq \mu .\nonumber
\end{align}
The last line is due to $\vert 1+d_i^t\vert\leq (1+\beta)\leq 2$.

{\bf Step 4.}  We would like to show that, if $\phi(x^{t+2})-\phi(x^{t+1})>-\frac{\mu}{16}$, then $(X_{t+1}\nabla^2f(x^{t+1})X_{t+1}+4\mu I)_{AX_{t+1}}\succeq 0$. To this end, we invoke \eqref{lambda} (where we let $t:=t+1$), \eqref{case 1 result 2nd} (where we let $t:=t+1$), and \eqref{second order condition per-iteration} (where we let $t:=t+1$). The combination of the three results gives 
\begin{align}
\left(X_{t+1}\nabla^2f(x^{t+1})X_{t+1}+\frac{\|p(x^{t+1},y^{t+1})\|}{\beta} I\right)_{AX_{t+1}}\succeq 0.\label{second order approximate second}
\end{align}
Further observe that from Step 2, it holds that, if $\phi(x^{t+2})-\phi(x^{t+1})>-\frac{\mu}{16}$, then $\frac{\|p(x^{t+1},y^{t+1})\|}{\beta}\leq \frac{\mu}{\beta}=4\mu$. Combined with \eqref{second order approximate second}, we have the claimed result in this step. The rest of the proof is straightforward following Step 5 of the proof for Theorem \ref{theorem 3.3}, while we let $\mu:=\frac{\varepsilon}{4}$ and $t^*:=\frac{64(f(x^0)-f^*+O(1)-\varepsilon)+1}{\epsilon}$.
\end{proof}

\begin{remark}
We notice the substantial improvement in the iteration complexity: If $f$ is quadratic, the complexity in achieving an $\varepsilon$-perturbed first-order and second-order stationary point is both $O(\varepsilon^{-1})$, while for the same algorithm to solve a more general problem, our complexity estimates are $O(\varepsilon^{-3/2})$ and $O(\varepsilon^{-3})$ for the first-order and second-order stationary points, respectively. The cause of this gap, to our understanding, is whether the cubic error term is present in the Taylor expansion-like inequalities \eqref{taylor expansion 1} and \eqref{taylor expansion 2}, or namely, whether $\eta=0$ holds. Note that when the $p$-th order derivative is used to find a first-order stationary point with a more general set of convex constraints, the best known iteration complexity is $O(\varepsilon^{-(p+1)/p})$ \cite{birginmp,birginsiam} (but with a costly per-iteration complexity). The quadratic case here discussed is compatible with this result as a limiting case $p\to+\infty$.
\end{remark}

\begin{remark}
In all three cases of discussion above, the per-iteration problem of the second-order ITRP admits a bisection scheme as per \cite{Ye98,ye1992affine} with a ``log-log'' (quadratic) rate of complexity.
\end{remark}

\section{Conclusion}
In this paper we consider the minimization of a continuous function that is potentially not differentiable or not twice-differentiable on the boundary of the feasible region. To characterize computable stationary points, we present suitable first- and second-order optimality conditions for this problem that generalizes to classical ones when the derivative on the boundary is available, through the use of an interior point technique. As a result, such an optimality condition is stronger than the existing conditions commonly used in the literature. We further develop new interior trust-region point algorithms and present  their worst-case complexity estimates to solve the special but important case with linear constraints. Even with a weaker regularity  on the objective function, the presented algorithms are theoretically guaranteed to yield a stronger optimality condition at the same best known complexity rates in the literature for first- and second-order stationarity using first- and second-order derivatives. We believe that this approach can be generalized for non-linear constraints and for infeasible initialization. Also, solving a higher-order subproblem, we believe this approach can yield iteration complexity results for finding $q$-th order stationary points, extending the results from \cite{beyond}.

\section*{Acknowledgement} This work was supported by the S\~ao Paulo Research Foundation (FAPESP grants 2013/05475-7 and 2016/02092-8) and the Brazilian National Council for Scientific and Technological Development (CNPq).
The content is solely the responsibility of the authors and does not necessarily represent
the official views of  the FAPESP and  CNPq. \jo{This research was conducted while the first author is holding a Visiting Scholar position at Department of Management Science and Engineering, Stanford University, Stanford CA 94305, USA.}

\bibliographystyle{plain}

\end{document}